\documentclass[smallextended,final,numbook,envcountsame]{svjour3}
\usepackage{times}
\usepackage[T1]{fontenc}   
\usepackage{amsmath}
\usepackage{amssymb}
\usepackage{graphics,graphicx}
\newcommand{\junk}[1]{}

\usepackage{subcaption}

\title{Efficiently Correcting Matrix Products
\thanks{A preliminary, short
version of this paper with weaker bounds
in the randomized setting has appeared in
\emph{Proceedings of 
25th International Symposium on Algorithms and Computation
(ISAAC 2014),  Lecture Notes in Computer Science,
Volume 8889, Jeonju, Korea, December 15-17, 2014.}}}
\author{
Leszek G\k{a}sieniec 
\and
Christos Levcopoulos
\and
Andrzej Lingas
\and
Rasmus Pagh
\and
Takeshi Tokuyama 
\institute{L. G\k{a}sieniec \at
Department of Computer Science, University of Liverpool, Ashton
Street, L69 38X, U.K.
\newline
\email{L.A.Gasieniec@liverpool.ac.uk}
\and 
C. Levcopoulos \and A. Lingas \at
Department of Computer Science, Lund University, 22100 Lund, Sweden. 
\newline
\email{Christos.Levcopoulos@cs.lth.se}
\email{Andrzej.Lingas@cs.lth.se}
\and
R. Pagh \at
Theoretical Computer Science section, IT University of Copenhagen, Denmark
\email{pagh@itu.dk}
\and T. Tokuyama \at
Graduate School of Information Sciences, Sendai,Tohoku University,
Japan.
\newline
\email{tokuyama@dais.is.tohoku.ac.jp}
}
}
\titlerunning{Efficiently Correcting Matrix Products}

\authorrunning{L. G\k{a}sieniec \emph{et al.}}
\date{}
\journalname{Algorithmica}

\begin{document}
\maketitle
\begin{abstract}
We study the problem of efficiently correcting an erroneous
product of two $n\times n$ matrices over a ring. Among other things,
we provide a randomized 
algorithm 
for correcting a matrix product with at most~$k$ erroneous entries
running in 
\newline
$\tilde{O}(n^2+kn)$ time
and a deterministic $\tilde{O}(kn^2)$-time algorithm for this
problem (where the notation $\tilde{O}$ suppresses polylogarithmic
terms in $n$ and~$k$).
\end{abstract}


\section{Introduction}

Matrix multiplication is a basic operation used in many scientific and
engineering applications. There are several potential reasons for erroneous
results of computation, in particular erroneous matrix products.  They
include software bugs, computational errors by logic circuits and
bit-flips in memory. Or, if the computation is performed by remote
computers or by parallel processors, some errors 
might be introduced due to faulty communication.

In 1977, Freivalds presented a randomized algorithm for verifying if a
matrix $C'$ is the matrix product of two $n\times n$ matrices $A$ and
$B,$ running in $O(n^2)$ time \cite{F77}.  His algorithm has been up
today one of the most popular examples showing the power of
randomization.

In spite of extensive efforts of the algorithmic community to
derandomize this algorithm without substantially increasing its time complexity,
one has solely succeeded partially, either decreasing the number of
random bits to a logarithmic one \cite{CK97,KS93,NN93} or using
exponentially large numbers and the unrealistic
BSS computational model \cite{KW14}.  One can argue that the latter
solutions in different ways hide additional $O(n)$ factors.
By the way, if one can use quantum devices then even an
$O(n^{5/3})$-time verification
of an $n\times n$ matrix product over an integral
domain is possible \cite{BS05}.

Interestingly, the problem of verifying matrix products
over the $(\min,+ )$ semi-ring seems to be much harder 
comparing to that
over an arbitrary ring.
Namely, it admits a truly subcubic algorithm if and only
if there is a truly subcubic algorithm for the all-pairs
shortest path problem on weighted digraphs (APSP) \cite{VW10}.

Freivalds' algorithm has also pioneered a new subarea of the so called
certifying algorithms \cite{MM11}.  Their purpose is to provide
besides the output a certificate or easy to verify proof that the
output is correct. The computational cost of the verification should
be substantially lower than that incurred by recomputing the output
(perhaps using a different method) from scratch.

In 1977, when Freivalds published his algorithm,
the asymptotically fastest known algorithm
for arithmetic matrix multiplication
was that due to Strassen running in $O(n^{2.81})$ 
time \cite{S69}. Since then the asymptotic running time
of fast matrix multiplication algorithms
has been gradually improved to $O(n^{2.3728639})$ 
at present \cite{CW,LG14,Vassilevska12} which is still substantially super-quadratic.

In this paper, we go one step further and consider
a more complex problem of not only
verifying a computational result but
also correcting it if necessary. Similarly
as Freivalds, as a subject of our study we
choose matrix multiplication.

Our approach is very different from that in fault tolerant setting,
where one enriches input in order to control the correctness of
computation (e.g., by check sums in the so called ABFT method)
\cite{DK11,WD11,WD14}.  Instead, we use here an approach resembling
methods from Combinatorial Group Testing where one keeps testing
larger groups of items in search for multiple targets, see,
e.g. \cite{DGV05,DH93}.

First, we provide a simple deterministic algorithm for correcting an
$n\times n$ matrix product $C'$ over a ring, with at most one
erroneous entry, in $O(n^2)$ time. It can be regarded as
a deterministic version of Freivalds' algorithm
(Section 3).  Next, we extend the
aforementioned algorithm to include the case when $C'$ contains at
most~$k$ erroneous entries. The extension relies on distributing
erroneous entries of $C'$ into distinct submatrices by
deterministically shuffling the columns of $C'$ and correspondingly
the columns of $B.$ The resulting
deterministic algorithm runs in $\tilde{O}(k^2n^2)$ time, where the
notation $\tilde{O}$ suppresses polylogarithmic terms in $n$ and~$k$
(Section 4). Then we show how to reduce the time bound to 
$\tilde{O}(kn^2)$ by applying this shuffling approach first with
respect to the columns and then with respect to the rows of $C'$.
In the same section,
we discuss also a slightly randomized version of the aforementioned
algorithm running in 
$\tilde{O}(\sqrt k n^2)$ expected time
using $O(\log^2k + \log k\log\log n)$
random bits.  
Next, in Section 5, 
we present a faster randomized algorithm for
correcting $C'$  in $O((n\sqrt {\log n}  
+\sqrt{k}\min\{k,n\})n\sqrt {\log n})$
time almost surely
(i.e., with probability at least $1-n^{-\alpha}$ 
for any constant $\alpha \ge 1$),
where~$k$ is the non-necessarily known number of  erroneous entries of
$C'$. A slight modification of this algorithm
runs in $O((n\log k +\sqrt k \min\{k,n\})n)$ expected time provided
that the number of erroneous entries is known.
This is our fastest algorithm for correcting $C'$ when
$k$ is very small.
Importantly, 
all our algorithms in Sections 3-5
are combinatorial (thus, they do not rely on the known
fast algorithms for matrix multiplication or fast polynomial multiplication)
 and easy to implement. In Section 6, we present a more advanced
algebraic approach based on the \emph{compressed matrix
multiplication} technique from~\cite{P13}.
In effect, we obtain a randomized algorithm
for correcting $C'$ in $O((n + k\log k\log \log k)n\log n)$ time
almost surely. Roughly, it asymptotically subsumes the randomized
algorithms of Section 5 for~$k$ larger than $n^{2/3}$ 
and asymptotically matches them up to a polylogarithmic factor  for
the remaining $k.$
We conclude with Final Remarks, where we discuss
how 
some of our randomized algorithms
can be also adjusted
to the situation when the number of erroneous entries
is unknown. 
For a summary of our results, see Table 1.

\begin{table*}[t]
\begin{center}
\begin{tabular}{||c|c|c||} \hline \hline
$\#$ errors $=e\le k$  & deterministic/randomized & ~time complexity~ 
\\ \hline \hline
$k=1$	& deterministic & $O(n^2)$ 
	\\ \hline
$k$ known & deterministic & $\tilde{O}(kn^2)$
\\ \hline
$k=e$ & $O(\log^2k + \log k\log\log n)$ & $\tilde{O}(\sqrt k n^2)$
\\
known & random bits & expected
\\ 
\hline
$k=e$ & randomized & 
$O((n\sqrt {\log n}+\sqrt{k}\min\{k,n\})n\sqrt {\log n})$ 
\\ 
unknown & & almost surely
\\\hline
$k=e$ & randomized & $O((n\log k +\sqrt k \min\{k,n\})n)$ 
\\
known &  & expected
\\\hline
$k$ known & randomized &  $O((n + k\log k\log \log k)n\log n)$
\\
 &  & almost surely 
\\ \hline \hline
\end{tabular}
\label{table: 1}
\vskip 0.5cm
\caption{The characteristics and time performances
of the algorithms for correcting an $n\times n$ matrix product
with at most~$k$ erroneous entries presented in this paper. 
The issue of adapting some of our randomized algorithms 
to unknown
$k$ is discussed in Final Remarks. }
\end{center}
\end{table*}

\section{Preliminaries}
Let $(U,+,\times)$ be a semi-ring.
 For two $n$-dimensional vectors 
$a=(a_0,...,a_{n-1})$ and $b=(b_0,...,b_{n-1})$ 
with coordinates in $U$ their dot product
$\sum_{i=0}^{n-1}a_i\times b_i $ 
over the semi-ring is denoted by  $a\odot b.$

For an $p\times q$ matrix $A=(a_{ij})$
with entries in $U,$ its
$i$-th row $(a_{i1},...,a_{in})$ is denoted by $A(i,*).$
Similarly, the $j$-th column $(a_{1j},...,a_{nj})$ 
of $A$ is denoted by $A(*,j).$ Given another
$q\times r$ matrix $B$ with entries in $U,$  the matrix
product $A\times B$ of $A$ with $B$ 
over the semi-ring is 
a matrix $C=(c_{ij})$, where $c_{ij}=A(i,*)\odot B(*,j)$
for $1\le i,j\le n.$ 

\section{Correcting a matrix product with a single error}

Given two matrices $A,\ B$
of size $p\times q$ and $q\times r,$
respectively, and
their possibly erroneous $p\times r$
matrix product $C'$ over a ring,
Freivalds' algorithm picks uniformly
at random a vector in $\{0,\ 1\}^r$
and checks if $A(Bx^T)=C'x^T$, where
$x^T$ stands for a transpose of $x$,
i.e., the column vector corresponding to $x$ \cite{F77}. 
For $i=1,...,p,$ if the $i$-th row
of $C'$ contains an erroneous entry,
the $i$-th coordinates of the vectors
 $A(Bx^T)$ and $C'x^T$ will differ
 with probability at least $1/2.$

In the special case, when $C'$ 
contains a single error, we can
simply deterministically set $x$
to the vector $(1,...,1)\in \{0,\ 1\}^r$ in the 
aforementioned Freivalds' test. The
vectors $A(Bx^T)$, $C'x^T$ will differ
in exactly one coordinate whose
number equals the number of the row
of $C'$ containing the single erroneous entry.
(Note that the assumption that there 
is only one error is crucial here 
since otherwise two or more errors
in a row of $C'$ potentially could
cancel out their effect so that the
dot product of the row with $x,$ which in this case
is just the sum of entries in the row,
would be correct.)
Then, we can simply compute the $i$-th
row of the matrix product of $A$ and $B$
in order to correct $C'$.

The time complexity is thus linear with respect
to the total number of entries in all three
matrices, i.e., $O(pq+qr+pr)$. More precisely, 
it takes time $O(p \cdot r)$ to compute  $C'x^T$, 
$O(q \cdot r)$ to compute  $Bx^T$, and finally
 $O(p \cdot q)$ to compute  the product of $A$ with $Bx^T$.

\begin{theorem}\label{lem: single}
Let $A,\ B, \ C'$ be three matrices of size $p\times q,$
$q\times r$ and $p\times r$, respectively, over a ring.
Suppose that $C'$ is different from the matrix
product~$c$ of $A$ and $B$ exactly in a single entry.
We can identify this entry and correct it
in time linear with respect to the total number of
entries, i.e., in $O(pq+qr+pr)$ time.
\end{theorem}

\section{Correcting a matrix product with at most~$k$ errors}

In this section, we shall repeatedly
use a generalization of the deterministic version of
Freivalds' test  applied to detecting single erroneous
entries in the previous section.

Let $A,\ B$ be two $n\times n$ matrices,
and let $C'$ be their possibly faulty product matrix
with at most~$k$  erroneous entries,
over some ring.
Let $C^*$ and $B^*$ denote matrices
resulting from the same permutation of columns
in the matrices $C'$ and $B$.

Similarly as in the previous section, the generalized
deterministic version of 
\newline
Freivalds' test
verifies rows of $C^*,$
but only for a selected set of consecutive columns
of the matrix. Such a set of columns  will be called a {\em strip}.

We shall check
each strip of $C^*$
independently for erroneous entries
that occur in a single column of the strip.
To do this, when we
determine the  vector $v$ to be used
in the coordinate-wise comparison
of $A(B^*v^T)$ with $C^*v^T$, we set the
$i$-th coordinate of $v$ to $1$ if and only
if the $i$-th column of the matrix $C^*$ 
belongs to the strip we
want to test. Otherwise, we set the coordinate  to $0.$
(See Fig.~\ref{fig: illustration}.)

\begin{figure}[hbtp]%
  \begin{center}%
    \includegraphics[scale=0.286,keepaspectratio]{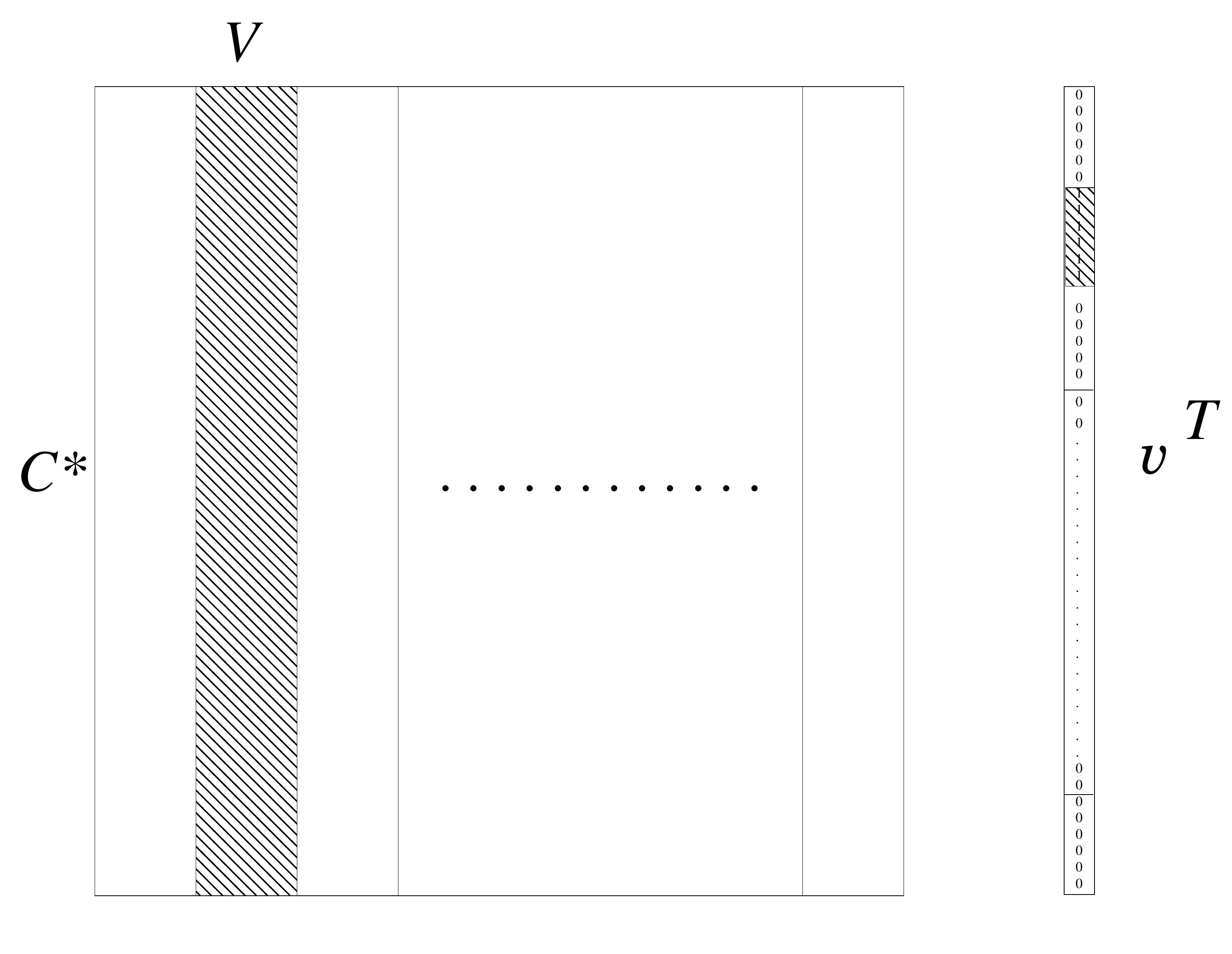}%
    \caption{Illustration of using the vector  $v^T$  in order to 
``extract'' the vertical strip  $V$ from the matrix  $C^*$.}%
    \label{fig: illustration}
  \end{center}%
\end{figure}

In this way,
for each row in a strip, we can detect whether or not the
strip row contains a single error.
The time complexity for testing a whole strip in this way
 is $O(n^2)$, independently from the number of columns of the strip.
If necessary, we can also correct a single row of a strip
by recomputing all its entries in time proportional 
to $n$ times the number of columns in the strip.

Our algorithm in this section relies also on the following 
number theoretical lemma.

\begin{lemma}\label{lem: prime}
Let $P=\{i_1,...,i_l\}$ be 
a set of $l$ different
indices in $\{1,...,n\}.$
There exists a constant~$c$ and for 
each $i_m\in P,$ a prime $p_m$
among the first $cl\log n / \log \log n$ primes
such that for $i_q\in P\setminus \{i_m\}$,
$i_m \bmod p_m \neq i_q \bmod p_m .$
\end{lemma}
\begin{proof}
It follows from the Chinese remainder theorem,
the density of primes
and the fact that each index in $P$ has $O(\log n)$
bits that there is a constant $b$ such that
for each pair $i_m,\ i_q$ of distinct indices
in $P$ there are at most 
$b\log n/\log \log n$ primes $p$
such that $i_m \bmod p = i_q \bmod p.$
Consequently, 
for each $i_m\in P$
there are at most 
$b(l-1)\log n / \log \log n$ primes $p$
for which there exists $i_q\in P\setminus \{i_m\}$ 
such that $i_q \bmod p
= i_m \bmod p .$
Thus, it is sufficient to set the constant~$c$
to $b$ in order to obtain the lemma.
\qed \end{proof}

\junk{Let $A,\ B$ be two $n\times n$ matrices,
and let $C'$ their possibly faulty product matrix
with at most~$k$  erroneous entries,
over some ring.}

Given the generalized deterministic version of
Freivalds' test and
Lemma \ref{lem: prime}, the idea of our algorithm
for correcting $C'$ is simple, see Fig.~\ref{fig: algo}.

For each prime $p$ among
the first $ck\log n/\log \log n$ primes, for $j=1,...,n,$
the $j$-th column is moved
into a (vertical) strip corresponding
to $j\bmod p.$ Correspondingly, the columns
of the matrix $B$ are permuted. 

\begin{figure}[htb]
\begin{footnotesize}
\noindent
{\bf Algorithm 1}
\par
\vskip 2pt
\noindent
{\em Input:} three $n\times n$ matrices $A,\ B,\ C'$
such that $C'$ differs from the matrix product
of $A$ and $B$ in at most~$k$ entries.
\par
\noindent
{\em Output:} the matrix product of $A$ and $B.$
\par
\vskip 2pt
\noindent
$L\leftarrow$ the set of the first  $ck\log n/\log \log n$ primes;
\par
\noindent
$C^*\leftarrow C';\ B^*\leftarrow B;$
\par
\noindent
{\bf for} each prime $p\in L$  {\bf do}
\begin{enumerate}
\item {\bf for} $j=1,...,n$ {\bf do}
\begin{enumerate}
\item Move the $j$-th column of $C^*$ into the
$j\bmod p +1$ strip of columns in $C^*;$
\item Correspondingly move the $j$-th 
column of $B^*$ into the
$j\bmod p +1$ strip of columns in $B^*;$
\end{enumerate}
\item {\bf for} each strip $V$ of $C^*$ {\bf do}
\begin{enumerate}
\item Set $v$ to the vector in $\{ 0,\ 1\}^n$
whose $j$-th coordinate is  $1$ if and only if
the $j$-th column of $C^*$ belongs to $V$;
\item Compute the vectors $A(B^*v^T)$ and $C^*v^T;$
\item {\bf for} each coordinate $i$ in which
$A(B^*v^T)$ and $C^*v^T$ are different {\bf do}
\begin{enumerate}
\item Compute the entries in the $i$-th row
of the strip of $A\times B^*$ corresponding
to $V$ and correct the $i$-th row of $V$ in~$c$
appropriately.
\end{enumerate}
\end{enumerate}
\end{enumerate}
Output $C^*.$
\end{footnotesize}
\caption{A deterministic algorithm for correcting
at most~$k$ errors}\label{fig: algo}
\end{figure}

Let $B^*$ and $C^*$ denote the resulting shuffled
matrices. 

Next, for each strip $V$
of $C^*$, we set $v$ to the vector
in $\{0,\ 1\}^n$ whose $j$-th coordinate is $1$
if and only if the $j$-th column belongs to $V.$
We compute and compare coordinate-wise the vectors
$A(B^*v^T)$ and $C^*v^T.$ 
Note that for $i=1,...,n,$ if there is a single
erroneous entry in the $i$-th row of $V$ then
the vectors $A(B^*v^T)$, $C^*v^T$ are different
in this coordinate. Simply, the $i$-th coordinate
of $C^*v^T$ is just the sum of the entries 
in the $i$-th row of $V$ while that coordinate of
$A(B^*v^T)$ is the sum of the entries in the 
$i$-th row of the vertical strip of the product
of $A$ and $B^*$ corresponding to $V.$

It follows in particular that for each strip which contains
only one erroneous column, we shall find all
erroneous rows in the strip.
Furthermore, we can correct all the erroneous
entries in a detected erroneous row
of the vertical strip $V$ in  $O(n^2/p)$ time
by computing $O(n/p)$ dot products
of rows of $A$ and columns of $B^*.$
Thus, in particular the correction of a single error
in a row of $V$ takes $O(n^2/p)$ time.

It follows from Lemma \ref{lem: prime},
that for each erroneous column in $C',$ there is
such a prime $p$ that the column is a single erroneous
column in one of the aforementioned vertical strips
of the shuffled matrix $C^*.$ 
Hence, all the~$k$ errors can be localized and corrected.

\begin{lemma} \label{lem: det}
Let $A,\ B, \ C'$ be three $n\times n$ matrices over a ring.
Suppose that $C'$ is different from the matrix
product~$c$ of $A$ and $B$ in at most~$k$ entries.
Algorithm 1 identifies these erroneous entries and corrects them
in $\tilde{O}(k^2n^2)$ time. 
\end{lemma}
\begin{proof}
The correctness of Algorithm 1
(see Fig.~\ref{fig: algo}) follows from 
the above discussion and 
Lemma \ref{lem: prime}.

Algorithm 1 iterates over $ck\log n/\log \log n$
smallest primes. Since an upper bound on the $i$-th prime
number is $O(i \log i)$ for any $i>1$, it follows that the largest prime considered
by the algorithm  has size $O(ck\log n\log k)$, and hence all these primes can be
listed in $O(c^2k^2\log^2 n\log k)$ time.

For a given prime $p,$ the algorithm
tests $p$ vertical strips $V$
for the containment of rows with single errors
by computing the vectors $A(B^*v^T)$ and $C^*v^T$.
It takes $O(n^2p)$ time. 

By the upper bounds on the number of considered primes
and their size, the total time taken by the tests
for all considered primes is 
$O(c^2k^2n^2\log^2 n \log k/\log \log n).$

The correction of an erroneous entry
in a detected erroneous row in a vertical strip
$V$ takes $O(n^2p)$ time. Hence, the correction
of the at most~$k$ erroneous entries 
in $C^*$ takes $O(ck^2n^2\log n \log k)$ time.

The tests and corrections dominate the running time
of the algorithm.
\qed \end{proof} 

In a practical implementation of the algorithm above,
one can of course implement the shuffling of the columns 
without actually copying data from one column to another.
For this purpose one could also define the strips in a different
way, i.e., they do not need to consist of consecutive columns.

\subsubsection{Reducing the time bound to $\tilde{O}(kn^2)$.}

In order to decrease the power of~$k$ in the upper bound
of the time complexity from $2$ to $1$, we make the following 
observation. Consider any column $i$ of $C'$. The number of erroneous 
entries in column $i$ that are in rows that have at least $\sqrt k$
erroneous  entries is at most $\sqrt k$.

We start by applying Algorithm 1 but only using the smallest
$c \sqrt{k} \log n / \log \log n$ primes. In this way
all rows that have at most $\sqrt k$ erroneous entries 
will be found in total
$\tilde{O}((\sqrt{k})^2 n^2)$ time, and will be fixed in
$O(n^2)$ time for each detected erroneous row.
So the time complexity up to this stage is dominated
by $\tilde{O}(kn^2)$.

Now, we let $C''$ be the partially corrected matrix and 
we apply the same procedure but
reversing the roles of columns and rows, i.e., we work with 
$ B^T A^T$ and $C''^T$. 
Since for any row of $C''^T$, all its erroneous entries that were in
columns of 
$C''^T$ with at most $\sqrt k$ errors were already corrected, 
now by the observation,
the number of erroneous entries in any row of $C''^T$ is at most $\sqrt
k$. Thus Algorithm 1 will now find all remaining erroneous rows in
time $\tilde{O}(kn^2)$ and we can correct them in additional
time $O(kn^2)$. Hence we obtain the following theorem:

\begin{theorem} \label{theo: det}
Let $A,\ B, \ C'$ be three $n\times n$ matrices over a ring.
Suppose that $C'$ is different from the matrix
product~$c$ of $A$ and $B$ in at most~$k$ entries.
We can identify these erroneous entries and correct them
in $\tilde{O}(kn^2)$ time. 
\end{theorem}

\subsubsection{Few random bits help.}
We can decrease the power of~$k$ in the upper bound
of Theorem \ref{theo: det} from $1$ to $0.5$
by using $O(\log^2k +\log k \log\log n)$ random bits
as follows and assuming that the exact number~$k$ of erroneous
entries in $C'$ is known. (The removal of this assumption
will be discussed later.) The idea is that instead
of testing systematically a sequence of primes, we
start by producing four times as many primes and then choose 
randomly among them in order to produce the strips.

We call a faulty entry in $C'$ 1-detectable if it lies in
a row or column of $C'$ with at most  $2 \sqrt k$ erroneous
entries. From this definition it follows that most faulty
entries are 1-detectable. More specifically, we call 
an entry in $C'$ 1-row-detectable, respectively 1-column-detectable,
if it lies in a row, respectively column, with at most  $2 \sqrt k$
erroneous entries. 

We will aim at detecting first a constant
fraction of the 1-row-detectable (false) entries,
and then a constant fraction of the 1-column-detectable 
entries. For this purpose we start by producing, in a
preprocessing phase,  the smallest 
$4c \sqrt k\log n /\log\log n$ primes (i.e.,
four times as many primes as we
did in the deterministic algorithm of Theorem \ref{theo: det}).

To detect sufficiently many 1-row-detectable entries we run
one iteration of Algorithm 1, with the difference that
we use a prime chosen randomly among the produced
$4c \sqrt k\log n /\log\log n$ smallest primes.
In this way, for each 1-row-detectable entry there is
at least a probability 1/2 that it will be detected.

Then we repeat once more this procedure but reversing 
the role of columns and rows, i.e., by working
with $ B^T A^T$ and $C'^T$. In this way for each
1-column-detectable entry there is at least a
probability 1/2 that it will be detected.

In this way, now each 1-detectable entry has been
detected with probability at least 1/2. By correcting
all these detected entries, we thus reduce the total
number of remaining false entries by an expected
constant fraction.

Thus we can set~$k$ to the remaining number of
false entries and start over again with the resulting, partially
corrected matrix $C'$. We repeat in this way until all
erroneous entries are corrected. 

The expected time
bound for the tests and 
corrections
incurred by the first selected primes  
dominate the
overall expected time complexity.
Note that the
bound is solely
$O(c \sqrt k n^2 \log n \log k).$

The number of
random bits needed to select such a random
prime is only  
\newline
$O(\log k + \log\log n).$
For a small $k,$
this  is much less than the logarithmic
in $n$ number of random bits
used in the best
known $O(n^2)$-time verification
algorithms
for matrix multiplication
obtained by a partial derandomization
of Freivalds' algorithm \cite{CK97,KS93,NN93}.

The overall number
of random bits, if we proceed
in this way and use fresh random
bits for every new selection of a
prime number has to be multiplied
by the expected number of the $O(\log k)$ iterations
of the algorithm. Thus, it becomes
$O(\log^2 k + \log k \log \log n).$ 

Hence, we obtain the following slightly randomized version
of Theorem \ref{theo: det}.

\begin{theorem} \label{theo: few}
Let $A,\ B, \ C'$ be three $n\times n$ matrices over a ring.
Suppose that $C'$ is different from the matrix
product~$c$ of $A$ and $B$ in exactly~$k$ entries.
There is a randomized algorithm
that identifies these erroneous entries
and corrects them 
in $\tilde{O}(\sqrt k n^2)$ expected time
using $O(\log^2k + \log k\log\log n)$
random bits.
\end{theorem}

If the number~$k$ or erroneous entries
is not known, then our slightly randomized method
can be adapted in order to estimate the number 
of erroneous columns.
Since similar issues arise
in connection to another randomized approaches presented
in the next chapters, we postpone this discussion to
Final Remarks. 


\section{A simple randomized approach}\label{sec:randomizedl}

In this section,
similarly as in the previous one, we shall 
apply the original and modified Freivalds' tests.
First, we apply repeatedly the original Freivalds'
test to the input $n\times n$ matrices $A,$ $B,$ and $C'$
and then to their transposes. These tests allow us 
to extract a submatrix $C_1$ 
which very likely contains all erroneous entries of $C'.$
Finally, we apply modified Freivalds' tests to
(vertical) strips of the submatrix $C_1$ of $C'.$

In contrast with the previous section, 
the tests are randomized. The modified test is just a restriction of
Freivalds' original randomized algorithm \cite{F77} to a strip
of $C_1$ that detects each erroneous row of a strip with probability
at least  $1/2$ even if a row contains more
than one erroneous entry.

More precisely, the vector $v$ used to test a strip
of $C_1$ by comparing $A_1(B_1v^T)$ with $C_1v^T$,
where $A_1$ and $B_1$ are appropriate submatrices of
$A$ and $B,$ is set as follows. 
Suppose that $C_1$ is an $q\times r$ matrix. 
For $j=1,...,r,$ the $j$-th coordinate
of $v$ is set to $1$ independently with probability $1/2$ if and only
if the $j$-th column of $C'$ belongs to the strip we
want to test, otherwise the coordinate is set to $0.$
In this way, for each row in the strip, the test
detects whether or not the
strip row contains an erroneous entry  with probability at least
$1/2$, even if the row contains more
than one erroneous entry.
The test for a whole strip takes
 $O(n^2)$ time, independently of the number of columns of the strip. 

Using the aforementioned tests, we shall prove the following theorem.

\begin{theorem}\label{lemma: sqrt-k}
Let $A$, $B$ and $C'$ be three $n \times n$ matrices
over a ring.
Suppose that $C'$ is different from the matrix
product~$c$ of $A$ and $B$ in~$k$ entries.
There is 
\newline
a randomized algorithm that transforms $C'$ into
the product $A \times B$  in
\newline  
$O((n\sqrt {\log n} +\sqrt{k}\min\{k,n\})n\sqrt {\log n})$
time almost surely
without assuming any prior knowledge of $k.$
\end{theorem}
\begin{proof}
Let us assume for the moment that~$k$ is known in advance (this assumption will
be removed later). Our algorithm (see Algorithm 2 in Fig.~\ref{fig: algo2}) will
successively correct the erroneous entries of $C'$ until $C'$ will
become equal to  $A \times B$ . 

Our algorithm consists of two main stages.
In the first stage, the standard Freivalds'
algorithm is applied iteratively
to $A,$ $B,$ $C'$ and then to the transposes
of these matrices in order to filter out
all the rows and all the columns of $C'$
containing erroneous entries almost
certainly. If the number of the aforementioned
rows or columns is less than $\log n$ (e.g., when $k<\log n$) then 
all the entries in the rows or columns of
the product $A \times B$ are computed and the algorithm halts.
The computation of the aforementioned
entries takes $O(\min\{k,n\}n^2)$ time in total.
Otherwise, a submatrix $C_1$ of $C'$ consisting
of all entries on the intersection of the
aforementioned rows and columns is formed. 
It has at most $\min\{k,n\}$ rows and at most $\min\{k,n\}$
columns. 

In the second stage,
we consider a partition of the columns of $C_1$ into
at most $\lceil \sqrt {\frac k {\log n}} \rceil$ strips of equal size,
i.e., consecutive groups of at most 
$\min\{k,n\}/\lceil
\sqrt{\frac k {\log n}}\rceil$ columns of $C_1.$ 
We treat each such strip separately and independently. 
For each strip, we apply our modification of Freivalds' test
$O(\log n)$ times. In this way, we can 
identify almost surely which rows of the
tested strip contain at least one error.
(Recall that
for each iteration and for each strip row, the chance of detecting an error,
if it exists, is at least $1/2$.)
Finally, for each erroneous strip row,
we compute the correct values for each one 
of its $O(\min\{k,n\}/ \sqrt{\frac k {\log n}})$  entries. 

\begin{figure}[htb]
\begin{footnotesize} 
\noindent
{\bf Algorithm 2}
\par
\vskip 2pt
\noindent
{\em Input:} three $n\times n$ matrices $A,\ B,\ C'$
such that $C'$ differs from the matrix product
of $A$ and $B$ in at most~$k$ entries.
\par
\noindent
{\em Output:} the matrix product of $A$ and $B$, almost surely.
\par
\vskip 2pt
\par
\noindent
Run Freivalds' algorithm $c \cdot \log n$ times on $A,\ B,\ C';$
\par
\noindent
Set $R$ to the set of indices of at most~$k$ rows
of $C'$ detected to be erroneous; 
\par
\noindent
If $\# R \le \log n$
then compute the rows of the product of $A$ and $B$ 
whose indices are in $R$, output the product
of $A$ and $B$, and stop;
\par
\noindent
Run Freivalds' algorithm $c \cdot \log n$ times on
$A^T,\ B^T,\ (C')^T;$
\par
\noindent
Set $L$ to the set of indices of at most~$k$ columns
of $C'$ detected to be erroneous;
\par
\noindent
If $\# L \le \log n$
then compute the columns of the product of $A$ and $B$ 
whose indices are in $L$, output the product
of $A$ and $B$, and stop;
\par
\noindent
Set $C_1$ to the submatrix of $C'$ consisting of all entries
occurring in the intersection of rows with indices in $R$ and 
columns with indices in $L;$
\par
\noindent
If $C_1$ is empty then return $C'$ and stop;
\par
\noindent
Set $A_1$ to the submatrix of $A$ consisting of all rows
with indices in $R;$
\par
\noindent
Set $B_1$ to the submatrix of $B$ consisting of all columns
with indices in $L;$
\par
\noindent
{\bf for} $i=1,..., \lceil \sqrt{\frac k {\log n}} \rceil $ {\bf do}
\begin{enumerate}
\item Run the
strip restriction of
Freivalds' algorithm $c \cdot \log n$ times on 
$A_1,$ $B_1$ and the $i$-th (vertical) strip of $C_1;$
\item For each erroneous strip row found in the  $i$-th (vertical) strip of $C'$, compute
      each entry of this strip row of $C_1$ and update $C'$ accordingly;
\end{enumerate}

Output $C'.$
\end{footnotesize}
\caption{A randomized algorithm for correcting
at most~$k$ errors\label{fig: algo2}}
\end{figure}

In each
iteration of the test in Step 1 in the algorithm, each erroneous
row in $C'$ will be detected with
a probability at least $1/2$. 
Hence, for a sufficiently large constant~$c$ (e.g.,
c=3) all erroneous rows of $C'$ will be detected almost surely within 
$c \cdot \log n$ iterations in Step 1.
Analogously, all erroneous columns of $C'$ will be detected almost surely within 
$c \cdot \log n$ iterations in Step 3.
It follows that all the erroneous entries of
$C'$ will belong to the submatrix $C_1$ consisting
of all entries on the intersection of the aforementioned
rows and columns of $C',$ almost surely. Recall that $C_1$ has at most
$\min\{k,n\}$ rows and at most $\min\{k,n\}$ columns.

Next, similarly, in Step 7  in the algorithm, each erroneous
row in each of the 
$\lceil \sqrt{\frac k {\log n}} \rceil $ 
strips of $C_1$ will be detected
almost surely.
If we use the straightforward method in order to compute the correct
values of an erroneous strip row, then it will take $O(n)$ time per
entry. Since each strip row of $C_1$ contains
$O(\min\{k,n\}/\sqrt {\frac k {\log n}})$ entries, the time
taken by a strip row becomes $O(n\min\{k,n\}/\sqrt {\frac k {\log n}})$. Since there are at
most~$k$ erroneous strip rows, the total time for correcting all the
erroneous strip rows in all strips of $C_1$
is $O(\sqrt k \min\{k,n\}n\sqrt {\log n})$. 
\junk{To estimate
the total time taken by 
the logarithmic number of applications of Freivalds' tests to
$A,$ $B,$ $C'$ in Step 1, to the transposes of
these matrices in Step 3, and by
the logarithmic number of applications
of the restrictions of Freivalds'
tests to the $O(\sqrt {\frac k {\log n}})$ vertical strips of $C_1$ in
Step 7, recall that $A_1$ has at most $\min\{k,n\}$ rows
and $n$ columns while $B_1$ has $n$ columns and at most
$\min\{k,n\}$ rows. This yields an upper
time bound of $O(n^2 \cdot \log n +\sqrt{k} \cdot
\min\{k,n\}n\cdot \sqrt {\log n})$ on the total time taken by the
tests.}

The total time taken by 
the logarithmic number of applications of Freivalds' tests to
$A,$ $B,$ $C'$ in Step 1 and to the transposes of
these matrices in Step 3 is $O(n^2\log n).$
To estimate
the total time taken by 
the logarithmic number of applications
of the restrictions of Freivalds'
tests to the $O(\sqrt {\frac k {\log n}})$ vertical strips of $C_1$ in
and matrices $A_1$ and $B_1$ in Step 7, recall that $A_1$ has at most $\min\{k,n\}$ rows
and $n$ columns, $B_1$ has $n$ rows and at most
$\min\{k,n\}$ columns, while $C_1$ has at most $\min\{k,n\}$
rows and columns.
Hence, in particular multiplications
of $C_1$ by the restricted test vectors take
$O(\min\{k,n\}\min\{k,n\}/\sqrt {k/\log n} \times 
\sqrt {k/\log n}\log n)$ time
in total, which is $O(\sqrt k \min\{k,n\}n\sqrt {\log n})$
since $k\ge \log n$ in the second stage.
Similarly, multiplications of $B_1$ by the restricted test vectors
take
\newline
$O(n\min\{k,n\}/\sqrt {k/\log n} \times 
\sqrt {k/\log n}\log n)$ time
in total, which is again 
\newline
$O(\sqrt k \min\{k,n\}n\sqrt {\log n}).$
Note that the $n$-coordinate vectors resulting from
multiplications of $B_1$ with the restricted test vectors are not
any more restricted and potentially each of their coordinates may
be non-zero. Therefore, the multiplications of $A_1$ with
the aforementioned vectors take $O(\min\{k,n\}n \times 
\sqrt {k/\log n}\log n)$ time
in total, which is $O(\sqrt{k} \cdot
\min\{k,n\}n\cdot \sqrt {\log n})$.
All this yields an upper
time bound of $O(n^2 \cdot \log n +\sqrt{k} \cdot
\min\{k,n\}n\cdot \sqrt {\log n})$ on the total time taken by the
tests in both stages..

In the second stage of
Algorithm 2, if we use, instead of the correct number~$k$ of
erroneous entries, a guessed number $k'$ which is larger than~$k$, 
then the time complexity becomes $O(n^2 \cdot \log n +\sqrt{k'} \cdot
\min\{k',n\}n\cdot \sqrt {\log n})$. 
This would be asymptotically fine as long as $k'$ is within a
constant factor of~$k$. On the other hand, if we guess $k'$
which is much smaller than~$k$, then the length of each erroneous
strip row in $C_1$ may become too large. For this reason,  first we have 
to find an appropriate size $k'$ for the strips to be used by our algorithm. 
For this purpose, we perform the first
stage of Algorithm 2, i.e., the logarithmic number of
original Freivalds' tests on the input matrices and their transposes.
Next, we set $k'$ to the maximum $k_0$ of the number of
erroneous rows and the number of erroneous columns reported
by the aforementioned tests, and a small constant, e.g., $4.$
Then, we
multiply our guess by $4,$ until we reach a good balance. More 
precisely, for each such guessed $k'$, without correcting any errors, 
we 
consider the partition of the
submatrix $C_1$ into $O(\sqrt{\frac {k'}{\log n}})$ 
strips, and apply our modified test to
each strip. As soon as we discover more than $k'$ erroneous 
strip rows in $C_1,$
we break the procedure without correcting any errors, 
and we start over with a four times larger
guess $k'$.

The aforementioned method of guessing $k'$ may result in at most
$O(\log k)$ wrong guesses until we achieve a good guess. 
Since we multiply our guess every time with $4,$ we obtain a
geometric progression of the estimated costs of subsequent trials.
In this way, the
upper bound on the asymptotic complexity 
of the whole algorithm 
but the time complexity of
the first logarithmic number of original Freivalds' test is dominated by that
of the iteration for the final $k'.$ In this iteration,
we test each strip $c \cdot\log n$ 
times in order to detect almost surely  all erroneous strip rows. 
\junk{
Note that when the number of erroneous entries
is at most the start value $k_0,$
our algorithm will
keep its first guess, i.e., $k'=k_0,$ and
the number
of detected erroneous rows is at most $k'.$
On the other hand,
by the definition of $k_0,$ $k \ge k_0/4$ 
holds almost surely.  
Next, let us consider
the case when $k>k_0.$
With respect to our current guess $k'$, the number
of detected erroneous rows lies then almost surely between
$k'/4$ and $k'$. 
Since each such an erroneous row in $C_1$ contains
$O(\min\{k',n\}\sqrt {\log n}/\sqrt {k'})$ entries, it can be recomputed in 
$O(\min\{k',n\}n\sqrt {\log n}/\sqrt {k'})$ time. 
Consequently, the total time complexity
 of correcting all the at most $k'$ erroneous
rows  becomes 
\newline
$O(\sqrt {k'} \min\{k',n\}n\sqrt {\log n})$. 
Since $k' \leq 4k$ holds, the theorem follows.}
\junk{Since $k' \leq 4k$ holds, we conclude that
the corrections of entries of $C_1$
in the second stage take $O(\sqrt {k} \min\{k,n\}n\sqrt {\log n})$
time. This yields the theorem.}
\qed
\end{proof}

Algorithm 2  in the proof of  Theorem \ref{lemma: sqrt-k} 
can be modified in order to achieve an expected time bound 
of  $O((n\log k +(\sqrt k \min\{k,n\})n)$ for correcting all errors, if
$k$ is known in advance. 

In the first stage, we perform only a single test for
the matrices $A,$ $B$ and a single test for
their transposes. Note that each erroneous
entry of~$c$ occurs with probability at least $\frac 12$
in a detected erroneous row of~$c$ as well as with probability
at least $\frac 12$ in a detected erroneous column of $C.$
Hence, an erroneous entry occurs with probability at least
$\frac 14$ in the resulting matrix $C_1.$ It follows
that the expected number of erroneous entries in $C_1$ 
is at least one fourth of those in $C.$

Next, we modify the second stage of Algorithm 2 as follows.
We set the number of vertical strips to $\lceil \sqrt k \rceil .$ 
Next, instead of applying the strip
restriction of Freivalds' algorithm $c \cdot \log n$ times for
each strip, we apply it only once for each strip
and correct all erroneous rows which we detect.
By counting how many errors we have corrected,
we compute how many errors remain. Then we
recurse in the same way on the partially corrected matrix~$c$ using
as a parameter this new number of errors which remain
to be corrected.

During each iteration of the algorithm, each 
remaining error in~$c$ will be detected
and corrected with probability at least $\frac 12 \times \frac 14$.  
Thus, the expected number of remaining
errors will decrease at least by the multiplicative factor
$\frac 78$ after each iteration. 
It follows that the expected number of iterations is
$O(\log k).$ Consequently,
the total cost of the tests in the first stage becomes
$O(n^2\log k).$ For the total time cost of tests and corrections
in the second stage, 
we obtain a geometric progression on the expected time complexity 
of each iteration, and so the total expected time complexity
is dominated by the time taken by the first iteration, which
is $O(\sqrt{k}\min\{k,n\}  \cdot n )$. Thus we obtain the following
theorem.

\begin{theorem}\label{theorem:sqrt-k-known}
Let $A,\ B, \ C'$ be three $n\times n$ matrices over a ring.
Suppose that $C'$ is different from the matrix
product~$c$ of $A$ and $B$ in exactly~$k$ entries.
There is a randomized algorithm
that identifies these erroneous entries
and corrects them in
$O((n\log k +\sqrt k \min\{k,n\})n)$ expected time.
\end{theorem}

\section{A fast algebraic approach}\label{sec:compressmm}

In this section we present a fast randomized
 algorithm that makes use of the \emph{compressed matrix multiplication} technique presented in~\cite{P13}.
We choose to give a self-contained and slightly simplified description because we do not need the full power of the framework of~\cite{P13}.

For integer parameters $s, t$ to be chosen later, the construction uses $t$ pairs of hash functions $g_\ell, h_\ell : \{1,\dots,n\}\rightarrow \{1,\dots,s\}$, with $\ell=1,\dots, t$, chosen independently from a strongly universal family of hash functions~\cite{CW79}.
We will make use of the following property:
\begin{lemma}\label{lem:coll}\cite{P13}

	For $(i_1,j_1), (i_2,j_2) \in \{1,\dots,n\}^2$ where $(i_1,j_1)\ne (i_2,j_2)$ we have
	$$\Pr\left[g_\ell(i_1)+h_\ell(j_1) = g_\ell(i_2)+h_\ell(j_2)\right] \leq 1/s \enspace .$$
\end{lemma}

Our algorithm first computes the following $t$ polynomials based on the matrices $A=(a_{ik})$, $B=(b_{kj})$, and $C'=(c'_{ij})$:
\begin{equation}\label{eq:polydef}
	p_\ell(x) = \sum_{k=1}^n \left(\sum_{i=1}^n a_{ik} x^{g_\ell(i)}\right) \left(\sum_{j=1}^n b_{kj} x^{h_\ell(j)}\right) - \sum_{i=1}^n \sum_{j=1}^n c'_{ij} x^{g_\ell(i)+h_\ell(j)},
\end{equation}
for $\ell=1,\dots,t$.
Multiplication of the polynomials corresponding to $A$ and $B$ is done
efficiently (over any ring) using the algorithm of Cantor and 
Kaltofen~\cite{CK91}, based on the original polynomial multiplication algorithm of Sch{\"o}nhage and Strassen~\cite{SS71}.

Let $p(x)_m$ denote the coefficient of $x^m$ in a polynomial $p(x)$. For each entry $i,j$ of $C'$ we assess the error term that must be added to $c'_{ij}$ as the \emph{majority element} of the sequence $p_\ell(x)_{g_\ell(i)+h_\ell(j)}$, $\ell=1,\dots,t$.
We will choose $s$ and $t$ such that with high probability the correction term (in most cases zero) appears more than $t/2$ times in the sequence.
If there is no such element for some entry $i,j$ the algorithm fails.

\subsection{Correctness}

Suppose $C=AB=(c_{ij})$ is the true matrix product.
Expanding the sum (\ref{eq:polydef}) and reordering the order of summation we get:
\begin{align*}
	p_\ell(x) & = \sum_{i=1}^n \sum_{j=1}^n \sum_{k=1}^n a_{ik} x^{g_\ell(i)}  b_{kj} x^{h_\ell(j)} - \sum_{i=1}^n \sum_{j=1}^n c'_{ij} x^{g_\ell(i)+h_\ell(j)}\\
	& = \sum_{i=1}^n \sum_{j=1}^n (c_{ij} - c'_{ij})\, x^{g_\ell(i)+h_\ell(j)} \enspace .
\end{align*}
This means that each coefficient of $p_\ell(x)$ is a sum of error terms:
\begin{align*}
	p_\ell(x)_m & = \sum_{\substack{i,j\\ g(i)+h(j)=m}} c_{ij} - c'_{ij} \enspace .
\end{align*}

Let $K\subseteq \{1,\dots,n\}^2$ be the set of positions of errors.
For $i^*,j^*\in \{1,\dots,n\}$:
\begin{equation}\label{eq:estimator}
p_\ell(x)_{g(i^*)+h(j^*)} = c_{i^* j^*} - c'_{i^* j^*} + \sum_{\substack{(i,j)\in K\backslash \{(i^*,j^*)\}\\ g(i)+h(j)=g(i^*)+h(j^*)}} c_{ij} - c'_{ij} \enspace .
\end{equation}
Lemma~\ref{lem:coll} states that $g(i)+h(j)=g(i^*)+h(j^*)$ holds with probability at most~$1/s$.
By a union bound the probability that the sum in (\ref{eq:estimator}) has at least one nonzero term is at most~$k/s$.
Choosing $s\geq 3k$ we get that $p_\ell(x)_{g(i^*)+h(j^*)} = c_{i^*j^*} - c'_{i^*j^*}$ with probability at least  $2/3$.
By Chernoff bounds this implies that after $t$ repetitions the probability that $p_\ell(x)_{g(i)+h(j)} = c_{ij} - c'_{ij}$ does not hold for at least $t/2$ values of $\ell$ is exponentially small in~$t$.
Choosing $t=O(\log n)$ we can achieve an arbitrarily small polynomial error probability in $n$ (even when summed over all entries $i$, $j$).

\subsection{Time analysis}

Strongly universal hash functions can be selected in constant time and space~\cite{CW79}, and evaluated in constant time.
This means that they will not dominate the running time.
Time $O(n^2+ns)$ is used to compute the polynomials 
$\sum_{i=1}^n a_{ik} x^{g_\ell(i)}$, $\sum_{j=1}^n b_{kj} x^{h_\ell(j)}$, and $\sum_{i=1}^n \sum_{j=1}^n c'_{ij} x^{g_\ell(i)+h_\ell(j)}$ in (\ref{eq:polydef}).
This can be seen by noticing that each entry of $A$, $B$, and $C'$ occur in one polynomial, and that there are $2n+1$ polynomials of degree~$s$.
Another component of the running time is the $tn$ multiplications of degree-$s$ polynomials, that each require $O(s \log s \log\log s)$ operations~\cite{CK91}.
Finally, time $O(tn^2)$ is needed to compute the correction term for each entry $i,j$ of $C'$ based on the sequence $p_\ell(x)_{g_\ell(i)+h_\ell(j)}$.
With the choices $s=O(k)$, $t=O(\log n)$ the combined number of
operations (algebraic and logical) is $O(n^2\log n + kn\log n\log
k\log \log k).$ 

\begin{theorem}\label{thm:compressedmm}
Let $A$, $B$ and $C'$ be three $n \times n$ matrices
over a ring.
Suppose that $C'$ is different from the matrix
product~$c$ of $A$ and $B$ in at most~$k$ entries.
There is a randomized algorithm that transforms $C'$ into
the product $A \times B$  in
\newline  
$O((n + k\log k\log \log k)n\log n)$ time, i.e., $\tilde{O}(n^2 + kn)$
time, almost surely.
\end{theorem}

While the above assumes prior knowledge of~$k$, we observe in Final
Remarks that this assumption can 
be removed with only a slight increase in running time.
Observe that the algorithm of Theorem~\ref{thm:compressedmm} needs
$O(t\log n)$ bits of space, which is $O(\log^2 n)$.
\junk{
We remark that the general idea of using linear sketches to compute compact summaries of matrix products may be useful in general for correcting matrix products.
For example, Iwen and Spencer~\cite{IS09} show that for complex-valued matrix products there is an efficiently computable linear sketch that allows recovery of the matrix product if the number of nonzeros in each column is bounded by roughly $n^{0.3}$.
Using linearity one can subtract the linear sketch for $C'$ to get the linear sketch of $AB-C'$, which has~$k$ nonzero entries.
If the number of nonzeros in each column of $AB-C$ is bounded by $n^{0.3}$, they can all be computed in time $n^{2+o(1)}$.
However, it is not clear for which rings this method will work, so while this is an interesting direction for future research we do not pursue it further here.}

\section{Final Remarks}


The majority of our randomized algorithms,
in particular that from Section 6,
can be efficiently adapted to the case when the number~$k$ of errors
is unknown,
proceeding similarly as in the proof of Theorem~\ref{lemma: sqrt-k}.
First, observe that using a parameter value $k'$ that is larger than~$k$ by a constant factor will yield the same guarantee on correctness and asymptotic running time.
This means that we can try geometrically increasing values of $k'$,
for example $k'=4^l$ for $l=1,2,3,\dots$ until the algorithm 
returns a correct answer within the stated time bound (using a suitably large constant in place of the big-O notation).
Correctness is efficiently checked using Freivalds' technique.
This technique increases the time bound by at most a factor $\log n$ compared to the case where~$k$ is known.
Furthermore, if $k\geq n \log n$ the time will be dominated by the last iteration, and we get time bounds identical to the case of known~$k$.
A similar approach can also be
used for refining the slightly randomized method of Theorem~\ref{theo:
  few} when the number of errors~$k$ is not known in advance. However,
if there is no knowledge at all concerning the number of errors, it
may be difficult to handle the case when no errors are detected: does
this happen because there are no errors at all, or because there are
too many errors and we chose a random prime from a too small range,
thus failing to isolate 1-detectable false entries? For this
reason, if there is no known useful upper bound on the remaining
number of errors, and we do not detect any errors during a series of
iterations, we may have to resort to some of the known algorithms
which test whether there are any errors at all~\cite{CK97,KS93,NN93}.
All such known algorithms running in time $O(n^2)$ may need a
logarithmic number of random bits, so if~$k$ is very small then this
may be asymptotically larger than the low number of random bits stated
in Theorem~\ref{theo: few}.

Note that the problem of correcting a matrix product is very
general. In the extreme case, when all entries of the matrix $C'$ may
be mistrusted, it includes the problem of computing the matrix product
$C$ from scratch. Also, when the matrix~$c$ is known to be sparse,
i.e., mostly filled with zeros, then we can set $C'$ to the all-zeros
matrix, and apply our matrix correction algorithms in order to obtain
output-sensitive algorithms for matrix multiplication (the number of
non-zero entries in~$c$ equals the number of erroneous entries in
$C'$).  They will be
slower than those known in the literature
based on fast rectangular matrix multiplication
\cite{AP09,LeGall12,HP98,L09}
(cf. \cite{WY14}). 

Finally, the general idea of using linear sketches to compute compact
summaries of matrix products may be useful in general for correcting
matrix products.  For example, Iwen and Spencer~\cite{IS11} show that
for complex-valued matrix products there is an efficiently computable
linear sketch that allows recovery of the matrix product if the number
of nonzeros in each column is bounded by roughly $n^{0.3}$.  Using
linearity one can subtract the linear sketch for $C'$ to get the
linear sketch of $AB-C'$, which has~$k$ nonzero entries.  If the
number of nonzeros in each column of $AB-C$ is bounded by $n^{0.3}$,
they can all be computed in time $n^{2+o(1)}$.  However, it is not
clear for which rings this method will work, so while this is an
interesting direction for future research we do not pursue it further
here.


\begin{acknowledgements}
We thank anonymous referees for
helping us to improve preliminary versions of this paper.

Christos Levcopoulos and Andrzej Lingas
were  supported in part by Swedish Research Council grant
621-2011-6179.
Takeshi Tokuyama was supported by JSPS Grant Scientific Research (B) 15H02665,  JSPS Scientific Research for Innovative Area 24106007, and JST  ERATO Kawarabayashi Big-Graph Project.
\end{acknowledgements}



\newcommand{\CIAC}{Italian Conference on Algorithms and Complexity}
\newcommand{\COCOON}{Annual International Computing Combinatorics Conference (COCOON)}
\newcommand{\COMPGEOM}{Annual ACM Symposium on Computational Geometry (SoCG)}
\newcommand{\ESA}{Annual European Symposium on Algorithms (ESA)}
\newcommand{\FOCS}{IEEE Symposium on Foundations of Computer Science (FOCS)}
\renewcommand{\FOCS}{IEEE FOCS}
\newcommand{\FSTTCS}{Foundations of Software Technology and Theoretical Computer Science (FSTTCS)}
\newcommand{\ICALP}{Annual International Colloquium on Automata, Languages and Programming (ICALP)}
\renewcommand{\ICALP}{ICALP}
\newcommand{\IPCO}{International Integer Programming and Combinatorial Optimization Conference (IPCO)}
\renewcommand{\IPCO}{IPCO}
\newcommand{\ISAAC}{International Symposium on Algorithms and Computation (ISAAC)}
\renewcommand{\ISAAC}{ISAAC}
\newcommand{\ISTCS}{Israel Symposium on Theory of Computing and Systems}
\newcommand{\JACM}{Journal of the ACM}
\newcommand{\LNCS}{Lecture Notes in Computer Science}
\newcommand{\MOR}{Mathematics of Operations Research}
\newcommand{\SICOMP}{SIAM Journal on Computing}
\newcommand{\SIJDM}{SIAM Journal on Discrete Mathematics}
\newcommand{\SODA}{Annual ACM-SIAM Symposium on Discrete Algorithms (SODA)}
\renewcommand{\SODA}{ACM-SIAM SODA}
\newcommand{\SPAA}{Annual ACM Symposium on Parallel Algorithms and Architectures (SPAA)}
\newcommand{\STACS}{Annual Symposium on Theoretical Aspects of Computer Science (STACS)}
\newcommand{\STOC}{Annual ACM Symposium on Theory of Computing (STOC)}
\renewcommand{\STOC}{ACM STOC}
\newcommand{\SWAT}{Scandinavian Workshop on Algorithm Theory (SWAT)}
\renewcommand{\SWAT}{SWAT}
\newcommand{\TCS}{Theoretical Computer Science}

\newcommand{\Proc}{Proceedings of the }
\renewcommand{\Proc}{Proc. }

\end{document}